\documentclass[oribibl]{llncs}

\usepackage{makeidx}  
\usepackage{amssymb}
\usepackage{graphicx, subfigure}
\usepackage{wrapfig}
\usepackage{float} 
\usepackage{enumitem}
\usepackage[left=2cm,top=3cm,right=2cm]{geometry} 
\usepackage{mdframed}
\usepackage{rotating} 

\graphicspath{{./graphics/}}

\newtheorem{observation}[question]{Observation}

\usepackage{hyperref}
\hypersetup{
colorlinks=true, 
linkcolor=blue, 
citecolor=green, 
urlcolor=webbrown, 
bookmarksopen=true,
pdftitle={Trilateration Research}, 
breaklinks=true, bookmarks=true,bookmarksnumbered,
pdfauthor={Alina Shaikhet}, 
pdfsubject={}, 
pdfkeywords={}, 
pdfcreator={pdfLaTeX}, 
pdfproducer={LaTeX with hyperref and ClassicThesis} 
}
\usepackage{bookmark}

\begin{document}

\pagestyle{headings}  

\title{Optimal Art Gallery Localization is NP-hard\thanks{This work was supported by NSERC.}}
\titlerunning{Optimal Art Gallery Localization is NP-hard}  
%
\author{Prosenjit Bose\inst{1} \and Jean-Lou De Carufel\inst{2} \and Alina Shaikhet\inst{1} \and Michiel Smid\inst{1}}
\authorrunning{P. Bose, J.-L. De Carufel, A. Shaikhet and M. Smid} 

\institute{School of Computer Science, Carleton University, Ottawa, Canada,\\
\email{\{jit, michiel\}@scs.carleton.ca, alina.shaikhet@carleton.ca},\\
\and
School of Electrical Engineering and Computer Science, U. of Ottawa, Canada,\\
\email{jdecaruf@uottawa.ca}}

\maketitle              

\begin{abstract}
\emph{Art Gallery Localization ($AGL$)} is the problem of placing a set $T$ of broadcast towers in a simple polygon $P$ in order for a point to locate itself in the interior. For any point $p \in P$: for each tower $t \in T \cap V(p)$ (where $V(p)$ denotes the \emph{visibility polygon} of $p$) the point $p$ receives the coordinates of $t$ and the Euclidean distance between $t$ and $p$. From this information $p$ can determine its coordinates. We study the computational complexity of $AGL$ problem. We show that the problem of determining the minimum number of broadcast towers that can localize a point anywhere in a simple polygon $P$ is NP-hard. We show a reduction from \emph{Boolean Three Satisfiability} problem to our problem and give a proof that the reduction takes polynomial time.
\end{abstract}

\section{Introduction}
\label{sec:introduction}
\pdfbookmark[1]{Introduction}{sec:introduction}

The art gallery problem was introduced by Victor Klee in 1973. He asked how many guards are sufficient to \emph{guard} the interior of a simple polygon having $n$ vertices. It has been shown by Chv{\'a}tal that $\lfloor n/3\rfloor$ guards are always sufficient and sometimes necessary~\cite{Chvatal197539}, and that such a set of guards can be computed easily~\cite{Fisk1978374}. However, such solutions are usually far from optimal in terms of minimizing the number of guards for a particular input polygon. Moreover, it was shown by Lee and Lin in~\cite{Lee:1986:CCA:13643.13657} that determining an optimal set of guards is NP-hard, even for simple polygons. Refer also to the book ``Art Gallery Theorems and Algorithms'' by O'Rourke~\cite{O'Rourke:1987:AGT:40599} that presents more detailed study of the topic. An overview of NP-hardness can be found in the book by Garey and Johnson~\cite{Garey:1979:CIG:578533}.

In our research we combine the art gallery problem with trilateration. \emph{Trilateration} is the process that determines absolute or relative locations of points by measurement of distances, using the geometry of the environment. Trilateration is not only interesting as a geometric problem, but it also has practical applications in surveying and navigation, including global positioning systems (GPS). GPS satellites carry very stable atomic clocks that are synchronized with one another and continuously transmit their current time and position. These signals are intercepted by a GPS receiver, which calculates how far away each satellite is based on how long it took for the messages to arrive. GPS receivers use trilateration to calculate the user's location, based on the information received from different GPS satellites.

By combining the art gallery problem with trilateration we address the problem of placing stationary broadcast \emph{towers} in a simple polygon $P$ in order for a receiving point $p$ (let us call it an \emph{agent}) to locate itself. Towers can be viewed as GPS satellites, while agents can be compared to GPS receivers. Towers are defined as points in a polygon $P$ which can transmit their coordinates together with a time stamp to other points in their visibility polygon. In our context, \emph{trilateration} is a process where the agent can determine its absolute coordinates from the messages the agent receives. The agent receives messages from all the towers that belong to its visibility polygon. Given a message from the tower $t$ the agent can determine its distance to $t$. 

In~\cite{AGLproblem17} we showed how to position at most $\lfloor 2n/3\rfloor$ towers inside $P$ and gave a localization algorithm that receives as input only the coordinates of the towers that can see an agent $p$ together with their distances to $p$. We also showed that $\lfloor 2n/3\rfloor$ towers are sometimes necessary. 

In this paper we show that the problem of determining the minimum number of broadcast towers that can localize a point anywhere in $P$ is NP-hard. Our solution is closely related to the NP-hardness of determining the minimum number of point guards for an $n$-edge simple polygon presented by Lee and Lin in~\cite{Lee:1986:CCA:13643.13657}. To prove the NP-hardness of their problem, Lee and Lin show a reduction from \emph{Boolean Three Satisfiability} (3SAT). However, they do not show that a polynomial number of bits represents the gadgets in their reduction. We use a similar reduction from 3SAT to show that our problem is NP-hard. Moreover, we provide a proof that the reduction takes polynomial time. We demonstrate that the number of vertices of the constructed polygon is polynomial in the size of the given input instance of $3SAT$ and that the number of bits in the binary representation of the coordinates of those vertices is bounded by a polynomial in the size of the input. Refer to Section~\ref{subsec:polynomial}.

In Section~\ref{sec:preliminaries} we give basic definitions and present some properties and observations. In Section~\ref{sec:NP-hard} we present our main results of NP-hardness. Section~\ref{sec:conclusion} contains conclusions and possible ideas for future research. 

\section{Preliminaries}
\label{sec:preliminaries}
\pdfbookmark[1]{Preliminaries}{sec:preliminaries}

Let $P$ be a simple polygon having a total of $n$ vertices on its boundary (denoted by $\partial P$). Two points $u, v \in P$ are \emph{visible to each other} if the segment $\overline{uv}$ is contained in $P$; we also say that $u$ \emph{sees} $v$. Note that $\partial P \subseteq P$ and that $\overline{uv}$ may touch $\partial P$ in one or more points. For $u \in P$ the \emph{visibility polygon} of $u$ (denoted $V(u)$) is the set of all points $q \in P$ that are visible to $u$. Note that $V(u)$ is a star-shaped polygon contained in $P$ and $u$ belongs to its \emph{kernel} (the set of points from which all of $V(u)$ is visible).

Let $T$ be a set of points in $P$. Elements of $T$ are called \emph{towers}. 
For any point $p \in P$: for each $t \in T \cap V(p)$ the point $p$ receives the coordinates of $t$ and can compute the Euclidean distance between $t$ and $p$, denoted $d(t,p)$. 

The set $T$ is said to \emph{trilaterate} the polygon $P$ if for every point $p \in P$ its absolute location can be identified given the information it receives from all towers in $V(p)$. By the \emph{map} of $P$ we denote the complete information about $P$ including the coordinates of all the vertices of $P$ and the vertex adjacency list. We assume in this paper that our localization algorithm knows the map of $P$ and the coordinates of all the towers in $T$.

Let an agent $p$ be a point in the interior of $P$, whose location is unknown. By $C(x,r)$ we denote the circle centered at $x$ with radius $r$. If only one tower $t$ can see $p$ then $p$ can be anywhere on $C(t,d(t,p)) \cap V(t)$, which may not be sufficient to identify the location of $p$ (unless the agent and the tower are at the the same location). Refer to Fig.~\ref{fig:example1}. Notice that one must know the map of $P$ to calculate $V(t)$. Let $L(u, v)$ be the line through points $u$ and $v$. If a pair of towers $t_1$ and $t_2$ can see $p$ then the location of $p$ can be narrowed down to at most two points $C(t_1,d(t_1,p)) \cap C(t_2,d(t_2,p)) \cap V(t_1) \cap V(t_2)$ (which are reflections of each other along $L(t_1, t_2)$). Refer to Fig.~\ref{fig:example2}. In this case we say that there is an \emph{ambiguity} along the line, since we cannot choose one location over the other without any additional information. However, if the map of $P$ is given (and thus we know $V(t_1)$ and $V(t_2)$) and if we place both towers on the same edge of $P$ in $kernel(P) \cap \partial P$ then the intersection $C(t_1,d(t_1,p)) \cap C(t_2,d(t_2,p)) \cap V(t_1) \cap V(t_2)$ is a single point (highlighted in red on Fig.~\ref{fig:example3}). 
\begin{figure}[h]
\centering
\subfigure[]{%
		\label{fig:example1}%
		\includegraphics[height=0.15\textheight]{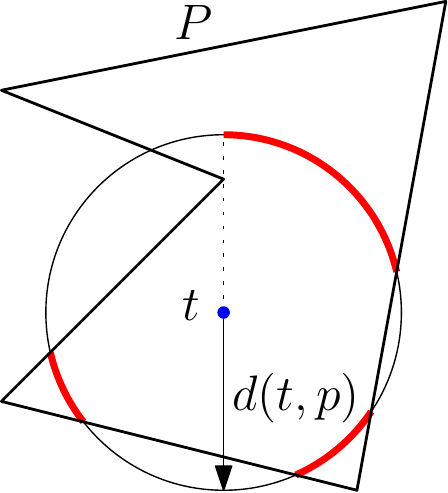}}
\hspace{0.055\textwidth}
\subfigure[]{%
		\label{fig:example2}%
		\includegraphics[height=0.15\textheight]{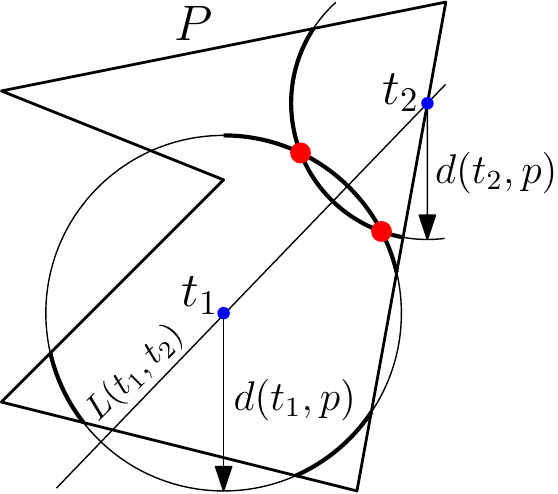}}
\hspace{0.03\textwidth}
\subfigure[]{%
		\label{fig:example3}%
		\includegraphics[height=0.15\textheight]{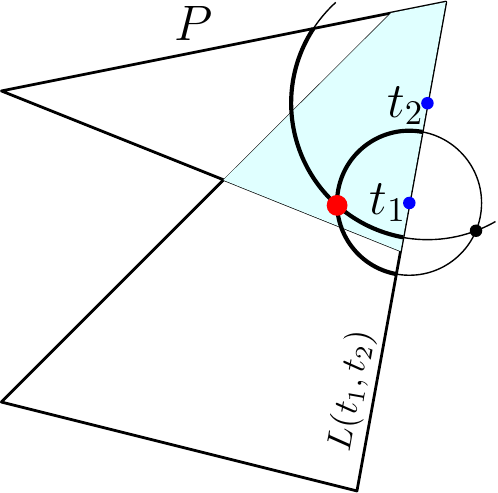}}
\hspace{0.05\textwidth}
\subfigure[]{%
		\label{fig:example4}%
		\includegraphics[height=0.15\textheight]{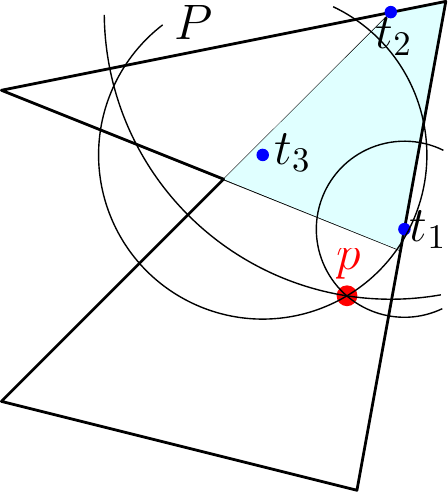}}
\caption{Trilateration example of a star-shaped polygon whose kernel (highlighted in cyan) does not degenerate into a single point. \textbf{(a)} The point $p$ can be anywhere on $C(t,d(t,p)) \cap V(t)$ (highlighted in red). \textbf{(b)} Ambiguity along the line $L(t_1, t_2)$; $p$ can be in one of the two possible locations (highlighted in red). \textbf{(c)} If the map of $P$ is given then the location of $p$ can be disambiguated and identified precisely. \textbf{(d)} The location of $p$ can be identified precisely without any knowledge about $P$.}
\label{fig:example}
\end{figure}
Notice, that in this case the kernel of $P$ must not degenerate into a single point, since we need at least two distinct points for tower placement. Alternatively, if the map of $P$ is \textbf{not} given, we can place a triple of non-collinear towers in the kernel of star-shaped $P$ (highlighted in cyan on Fig.~\ref{fig:example4}) to localize any point interior to $P$.

In~\cite{AGLproblem17} we provide a polynomial time algorithm for trilaterating a simple polygon $P$ in general position using at most $\lfloor 2n/3\rfloor$ towers. We do it by partitioning $P$ into at most $\lfloor n/3\rfloor$ star-shaped polygons $P_1, P_2, \ldots P_k$ such that: $P_1 \cup P_2 \cup \ldots \cup P_k = P$ and $kernel(P_i)$, for every $1 \leq i \leq k$, does not degenerate into a single point; and then we assign a pair of towers to each partition.

The optimization problem we study is the following. Given a simple polygon $P$, compute a set of towers $T$ of minimum size that trilaterates $P$. 

\section{Optimal Art Gallery Localization is NP-hard}
\label{sec:NP-hard}
\pdfbookmark[1]{Optimal Art Gallery Localization is NP-hard}{sec:NP-hard}

We want to show that determining an \emph{optimal solution} (a smallest set of towers) for trilaterating a polygon is NP-hard. To do this, we need to give a reduction from a known NP-hard problem. Similarly to the approach used by Lee and Lin presented in~\cite{Lee:1986:CCA:13643.13657}, we reduce $3SAT$ to our problem.

\medskip
\begin{mdframed}
\textbf{Art Gallery Localization ($AGL$) Problem:}

\emph{Instance}: A simple polygon $P$ of size $n$ and a positive integer $K$. 

\emph{Question}: Does there exist a set $T$ that trilaterates $P$ with $|T| \leq K$? 
\end{mdframed}

\begin{theorem}
\label{theo:NP}
The art gallery localization problem is $NP$-hard.
\end{theorem}

To show $NP$-hardness we reduce the following $NP$-complete problem to the $AGL$ problem.
 
\begin{mdframed}
\textbf{Boolean Three Satisfiability ($3SAT$) problem:}

\emph{Instance}: A set $U = \{u_1, u_2, \ldots, u_n\}$ of Boolean variables
and a collection $C = \{c_1, c_2, \ldots, c_m\}$ of clauses over
$U$ exist such that $c_i \in C$ is a disjunction ($OR$ or $\vee$) of precisely three
literals (where a literal is either a variable, or the negation of a variable).

\emph{Question}: Does there exist a truth assignment to the $n$ variables in $U$ such that the conjunctive normal form (CNF) $c_1 \wedge c_2 \wedge \ldots \wedge c_m$, evaluates to \emph{true}? 
\end{mdframed}

\medskip
We need to show that $3SAT$ can be transformed into $AGL$ in polynomial time. In other words, the goal is to transform a given instance of $3SAT$ into a simple polygon with $O((nm)^d)$ vertices (for some constant $d$) that can be trilaterated with $K$ or fewer towers \textbf{if and only if} the $3SAT$ instance is satisfiable. Let the bound $K$ for the $AGL$ problem be $8m + 2n + 2$.

We show how to construct a simple polygon in a step-by-step manner by describing the basic components from which the desired polygon is built. We describe a construction for literals, clauses (that contain several literals) and variables (where the consistency of \emph{true}/\emph{false} settings of the literals will be enforced). We want to construct a simple polygon such that no two different constructions can be completely visible to the same pair of towers. We achieve this by gluing together star-shaped polygons whose kernels do not have more than a single point in common.

\subsection{Literal Pattern}
\label{subsec:literal}
\pdfbookmark[2]{Literal Pattern}{subsec:literal}

The construction $P_l$ for a literal $l$ is the five-sided polygonal shape shown in Fig.~\ref{fig:literal1}, which has the property that the vertices $a$, $d$ and $c$ are collinear (refer to Fig.~\ref{fig:literal2}). The construction $P_l$ is attached to the body of the main polygon via a line segment $\overline{e a}$. This connection is shown as a blue dashed line in Fig.~\ref{fig:literal}.

\begin{figure}[h]
\centering
\subfigure[]{%
		\label{fig:literal1}%
		\includegraphics[width=0.114\textwidth]{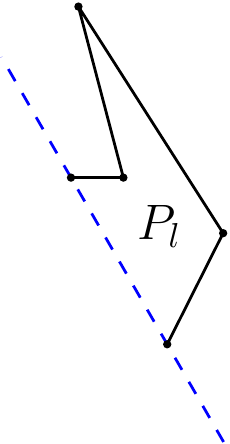}}
\hspace{0.05\textwidth}
\subfigure[]{%
		\label{fig:literal2}%
		\includegraphics[width=0.14\textwidth]{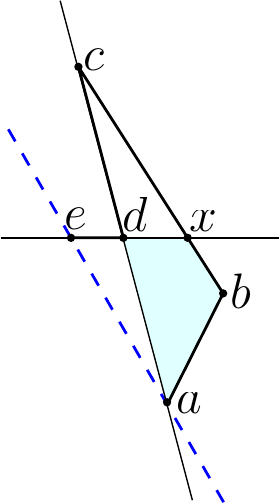}}
\hspace{0.05\textwidth}
\subfigure[]{%
		\label{fig:literal3}%
		\includegraphics[width=0.14\textwidth]{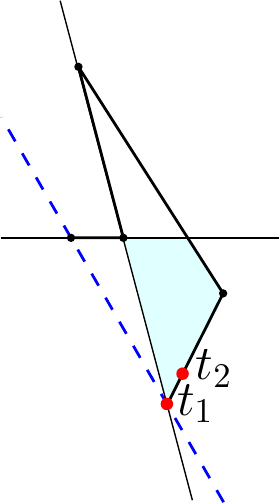}}
\hspace{0.05\textwidth}
\subfigure[]{%
		\label{fig:literal4}%
		\includegraphics[width=0.14\textwidth]{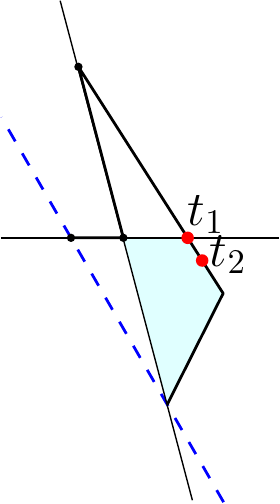}}
\caption{Construction $P_l$ of a literal $l$ (with the kernel highlighted in cyan). The blue dashed line indicates where this construction is attached to the main polygon. \textbf{(b)} The vertices $a$, $d$ and $c$ are collinear. Notice that $x$ is not a vertex of $P_l$. \textbf{(c)} The literal is assigned \emph{true} or  \textbf{(d)} \emph{false} value.}
\label{fig:literal}
\end{figure}

The polygonal region $P_l$ is a star-shaped pentagon whose kernel is not a single point. Moreover, $kernel(P_l) \cap \partial P_l$ contains two distinct points that belong to the same edge of $P_l$. Thus $P_l$ can be trilaterated with as few as two towers since the map of the polygon is given. However, these towers must be in $kernel(P_l)$ and positioned on the same edge of $P_l$. Let $x$ be a point where the extension of the edge $\overline{e d}$ intersects $\overline{c b}$. Refer to Fig.~\ref{fig:literal2}. It follows, that both towers must either be on the line segment $\overline{x b}$ or on the edge $\overline{a b}$.

Notice that a point can see the complete interior of the triangle $\triangle abc$ only if it is in the triangle. In other words, no tower from the outside of $\triangle abc$ can see its complete interior.

Depending on the truth value assigned to the literal $l$, the position of towers for $P_l$ trilateration will be enforced by elements of our construction, which will be described later in Sections~\ref{subsec:clause} --~\ref{subsec:complete}. For now, just note that if the literal $l$ is assigned \emph{true} (respectively \emph{false}) then the towers $t_1$ and $t_2$ are positioned at $\overline{a b}$ (respectively $\overline{x b}$). In particular, $t_1$ is positioned at vertex $a$ (respectively, at point $x$) and $t_2$ is positioned very close to $t_1$ on the corresponding line segment. Refer to Fig.~\ref{fig:literal3} for the \emph{true} assignment and to Fig.~\ref{fig:literal4} for the \emph{false} assignment.

\subsection{Clause Junction}
\label{subsec:clause}
\pdfbookmark[2]{Clause Junction}{subsec:clause}

\begin{wrapfigure}{r}{0.48\textwidth}
\vspace{-20pt}
\centering
\includegraphics[width=0.47\textwidth]{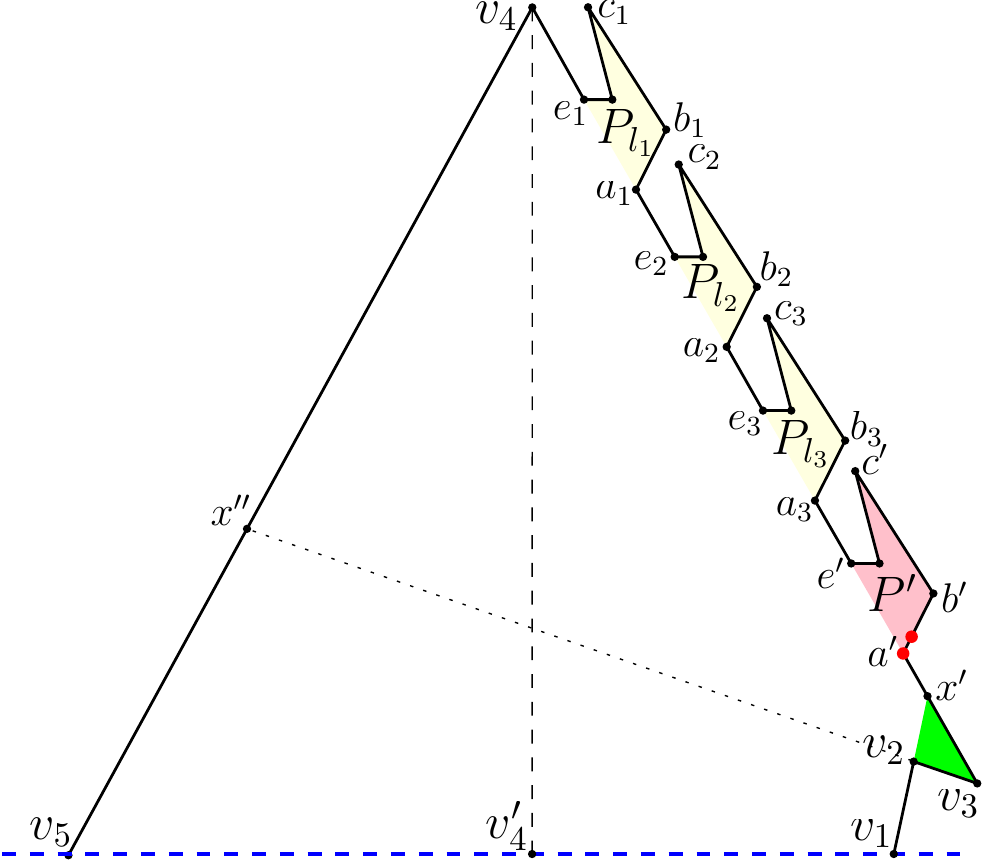}
\caption{Construction $P_C$ of a clause $C$. The blue dashed line indicates where this construction is attached to the main polygon. Fixed towers are highlighted in red.}
\label{fig:clause}
\vspace{-40pt}
\end{wrapfigure}
A construction $P_C$ for each clause $C$ will contain four polygonal regions that represent literals. Three of them will correspond to the literals of the given clause. The role of the fourth polygonal region of a literal form will be explained later in this subsection. Note for now that it will always be assigned a value of \emph{true}.

Consider the clause $C = l_1 \vee l_2 \vee l_3$, where $l_1 \in \{u_i, \overline{u_i}\}$, $l_2 \in \{u_j, \overline{u_j}\}$ and $l_3 \in \{u_k, \overline{u_k}\}$ are literals, and $u_i$, $u_j$ and $u_k$ are variables, $u_i, u_j, u_k \in U$. The basic construction for the clause $C$ is shown in Fig.~\ref{fig:clause}. It contains one subpolygon for each literal: $P_{l_1} = \{a_1, b_1, c_1, d_1, e_1\}$, $P_{l_2} = \{a_2, b_2, c_2, d_2, e_2\}$ and $P_{l_3} = \{a_3, b_3, c_3, d_3, e_3\}$; and an additional subpolygon $P' = \{a', b', c', d', e'\}$ of a literal form. The vertices $v_3$, $a'$, $e'$, $a_3$, $e_3$, $a_2$, $e_2$, $a_1$, $e_1$ and $v_4$ are collinear, i.e. they belong to $L(v_3, v_4)$. 

\begin{observation}
\label{observ:clause}
Every optimal solution for trilaterating $P_C$ consists of $8$ distinct towers with at least two of them positioned at $a_1$, $a_2$, $a_3$ or $a'$. 
\end{observation}
\begin{proof}
We showed in Section~\ref{subsec:literal} that each construction $P_l$ requires at least two towers in its kernel. Notice that no tower can see the complete interior of $P_l$ if it does not belong to $P_l$. Since no two pentagons $P_{l_1}$, $P_{l_2}$, $P_{l_3}$ or $P'$ intersect, each of them requires at least two distinct towers. It follows, that we cannot trilaterate $P_C$ with less than $8$ towers, meaning that every optimal solution for trilaterating $P_C$ will be of size at least $8$.

Notice, that $P_C$ is a union of pentagons $P_{l_1}$, $P_{l_2}$, $P_{l_3}$, $P'$ and $P^* = \{v_1, v_2, v_3, v_4, v_5\}$. Let $x'$ be a point where the extension of the edge $\overline{v_1 v_2}$ intersects $\overline{v_3 a'}$. Let $x''$ be an intersection point between the extension of the edge $\overline{v_2 v_3}$ and $\overline{v_4 v_5}$. Notice, that $\{v_2, x', v_4, x''\}$ is the kernel of $P^*$. We can trilaterate $P^*$ with as few as two towers either on $\overline{x' v_4}$ or $\overline{x'' v_4}$. Observe, that $a_1$, $a_2$, $a_3$ and $a'$ belong to $\overline{x' v_4}$. Thus, if any of the pentagons $P_{l_1}$, $P_{l_2}$, $P_{l_3}$ and $P'$ have towers  at $a_1$, $a_2$, $a_3$ or $a'$ then those towers can be reused for trilaterating $P^*$. 

It is left to show that there are no ambiguities along $L(v_3, v_4)$. Assume that only a pair of towers $t_1$ and $t_2$ (positioned at $a_1$, $a_2$, $a_3$ or $a'$) can see the agent with $p_1$ and $p_2$ being agent's possible locations. Assume, without loss of generality, that $p_1 \in P^*$. Recall, that we know the map of the polygon and thus we can see if $p_2$ is inside or outside of $P_C$. If $p_2 \notin P_C$ then the agent is at $p_1$. Otherwise, if $p_2 \in P_C$, then the agent is in one of the pentagons $P_{l_1}$, $P_{l_2}$, $P_{l_3}$ or $P'$. Recall, that we know the locations of all the towers and, together with the map of the polygon, we know visibility polygons of all the towers. We can deduce that $p_2$ does not belong to at lest one of $V(t_1)$ or $V(t_2)$. Thus the agent is at $p_1$.

It follows, that every optimal solution for trilaterating $P_C$ has size $8$ with towers positioned at at least two of the four vertices $a_1$, $a_2$, $a_3$ and $a'$. 
\qed
\end{proof}

Notice, that by construction, any tower that is located outside of $P_C$ (but inside $P \setminus P_C$) can only see a proper subset of $\triangle v_2 v_3 x'$ (highlighted in green in Fig.~\ref{fig:clause}). 

Observation~\ref{observ:clause} allows us to prove the following lemma. But first note, that the \emph{true}/\emph{false} position of towers in $P_{l_1}$, $P_{l_2}$ and $P_{l_3}$ is enforced from the outside of $P_C$ by other elements of our construction, which we describe later in Sections~\ref{subsec:variable},~\ref{subsec:complete}.

\begin{lemma}
\label{lem:clause}
A construction $P_C$ of a clause $C$ can be trilaterated with $8$ towers \textbf{if and only if} a truth value of $C$ evaluates to \emph{true}.
\end{lemma}
\begin{proof}

\textbf{($\rightarrow$)} 
Let $T$ be a set of towers of size $8$ that trilaterates $P_C$. By observation~\ref{observ:clause} there exists a tower at $a_1$, $a_2$ or $a_3$. It follows, that the corresponding construction $P_l$ has towers in a \emph{true} position, implying that the literal $l$ is assigned \emph{true}. Since $l$ is \emph{true}, the truth value of $C$ is \emph{true}.


\textbf{($\leftarrow$)}
Assume that a truth value of $C$ evaluates to \emph{true}. Therefore, at least one of the literals $l_1$, $l_2$ or $l_3$ is \emph{true}, and thus there exists a tower $t_1$ at $a_1$, $a_2$ or $a_3$. Let us trilaterate $P'$ with a pair of towers in \emph{true} position, meaning that we place tower $t_2$ at $a'$. Each of $P_{l_1}$, $P_{l_2}$, $P_{l_3}$ and $P'$ has two distinct towers. It is $8$ towers total (including $t_1$ and $t_2$). This tower set trilaterates $P_C$ (refer to the proof of observation~\ref{observ:clause} for details).

\qed
\end{proof}

\subsection{Variable Pattern}
\label{subsec:variable}
\pdfbookmark[2]{Variable Pattern}{subsec:variable}

We need to create a construction that will force all truth assignments of literals of a particular variable to be consistent with one another. A variable pattern serves this purpose. An example of the variable pattern $P_i$ for variable $u_i$ is given in Fig.~\ref{fig:variable2} and at the bottom half of Fig.~\ref{fig:variable1}. One such pattern will exist per variable in the final construction. 

\begin{figure}[h]
\centering
\includegraphics[width=0.7\textwidth]{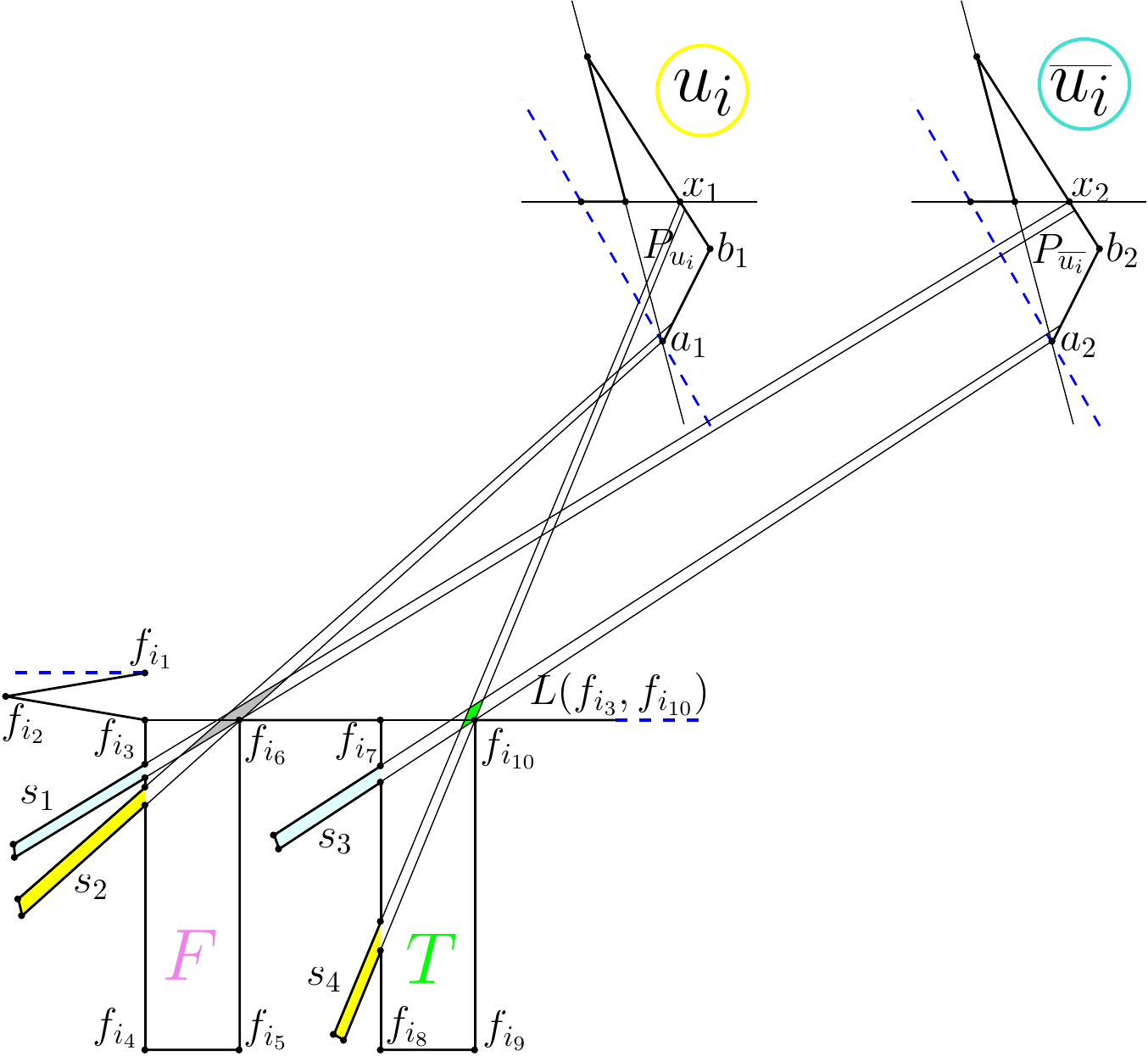}
\caption{An example of a construction $P_i$ for the variable $u_i$. The blue dashed lines indicate where this construction is attached to the main polygon. Literal patterns $P_{u_i}$ and $P_{\overline{u_i}}$ may belong to the same clause or to different clauses. Spikes of the same color are aligned with the same literal pattern.}
\label{fig:variable1}
\end{figure}

Every variable pattern contains two wells. The leftmost well (marked with $F$ in Fig.~\ref{fig:variable1}) will have one tower at the vertex $f_{i_6}$ and another in the close proximity to it (we will discuss later its specific location) if the given variable is assigned the truth value \emph{false}. Let us call this an \emph{F-well}. Similarly, if the variable is assigned the truth value \emph{true} then the rightmost well (marked with $T$ in Fig.~\ref{fig:variable1}) will have one tower at the vertex $f_{i_{10}}$ and another close to it. Let us call this a \emph{T-well}. For simplicity we depicted wells here as vertical rectangles. However, their boundaries are not parallel to each other and in the complete construction are aligned with a pair of vertices whose construction will be defined in Section~\ref{subsec:complete}.

The wells contain \emph{spikes}. The number of spikes each well contains in the final construction is equal to the number of times the corresponding variable or its negation appears in the given $3CNF$ formula. Suppose variable $u_i$ appears in clause $C$. The spikes in the F-well (respectively the T-well) for $u_i$ are aligned with the vertex $f_{i_6}$ (respectively, $f_{i_{10}}$) (and partially with its immediate neighbourhood). 

If $C$, for example, contains the literal $l_1 = u_i$ then the spike in the F-well (respectively, the T-well) that corresponds to $u_i$ in $C$ will also be aligned with $a_1$ (respectively, $x_1$) and its immediate neighbourhood on $\overline{a_1 b_1}$ (respectively, $\overline{x_1 b_1}$) of the literal pattern $P_{l_1}$ in clause junction $P_C$. Those spikes are $s_2$ and $s_4$ and they are highlighted in yellow in Fig.~\ref{fig:variable1}. 

If $C'$, for example, contains the literal $l_2 = \overline{u_i}$ then the spike in the F-well (respectively, the T-well) that corresponds to $u_i$ in $C'$ will also be aligned with $x_2$ (respectively, $a_2$) and its immediate neighbourhood on $\overline{x_2 b_2}$ (respectively, $\overline{a_2 b_2}$) of the literal pattern $P_{l_2}$ in clause junction $P_{C'}$. Those spikes are $s_1$ and $s_3$ and they are highlighted in cyan in Fig.~\ref{fig:variable1}.

Let us discuss ways in which a variable pattern can be trilaterated. Notice, that $P_i$ is a union of two star-shaped polygons $P_i^F = \{f_{i_1}, \ldots f_{i_6}, x_{i_1}\} \cup \{$spikes of the F-well$\}$ and $P_i^T = \{f_{i_1}, f_{i_2}, f_{i_3}, f_{i_7}\ldots f_{i_{10}}, x_{i_2}\} \cup \{$spikes of the T-well$\}$. Refer to Fig.~\ref{fig:variable2}. The kernel of the star-shaped polygon formed by the union of strips defined by the spikes in the F-well (respectively, the T-well) is highlighted in gray (respectively, green) and denoted $V(S_i^F)$ (respectively, $V(S_i^T)$). It is always possible to construct a variable pattern in a way such that $V(S_i^F) \cap V(S_i^T) = \emptyset$. This can be done by moving a T-well further away from an F-well, i.e. by enlarging the distance between $f_{i_6}$ and $f_{i_7}$. This ensures that there does not exist a star-shaped polygon that contains a variable pattern as its subpolygon (connected via the blue dashed lines shown in Fig.~\ref{fig:variable2}). This means that the trilateration of a variable pattern cannot be done with less than 4 towers.

The vertices $f_{i_3}$, $f_{i_6}$, $f_{i_7}$ and $f_{i_{10}}$ are collinear and belong to $L(f_{i_3}, f_{i_{10}})$. Let $L'$ be the line parallel to $L(f_{i_3}, f_{i_{10}})$ such that $f_{i_1} \in L'$. Let $x_{i_1}$ (respectively, $x_{i_2}$) be an intersection point of $L'$ and $L(f_{i_5}, f_{i_6})$ (respectively, $L(f_{i_9}, f_{i_{10}})$). The polygons $P_i^F = \{f_{i_1}, \ldots f_{i_6}, x_{i_1}\} \cup \{$spikes of the F-well$\}$ and $P_i^T = \{f_{i_1}, f_{i_2}, f_{i_3}, f_{i_7}\ldots f_{i_{10}}, x_{i_2}\} \cup \{$spikes of the T-well$\}$ are star-shaped by construction. Notice that $f_{i_6} \in kernel(P_i^F)$ and $f_{i_{10}} \in kernel(P_i^T)$. To trilaterate $P_i^F$ we position a pair of towers $t_F'$ and $t_F''$ in $kernel(P_i^F)$. In particular, we position $t_F'$ at $f_{i_6}$ and $t_F''$ in close proximity to $t_F'$ in $L(f_{i_5}, f_{i_6}) \cap kernel(P_i^F)$. In a similar way, to trilaterate $P_i^T$ we position $t_T'$ at $f_{i_{10}}$ and $t_T''$ in close proximity to $t_T'$ in $L(f_{i_9}, f_{i_{10}}) \cap kernel(P_i^T)$. The locations of the four towers are highlighted in red in Fig.~\ref{fig:variable2}. Notice that both pairs of towers $t_F'$, $t_F''$ and $t_T'$, $t_T''$ can see $\triangle f_{i_1} f_{i_2} f_{i_3}$.

\begin{wrapfigure}{r}{0.62\textwidth}
\vspace{-10pt}
\centering
\includegraphics[width=0.60\textwidth]{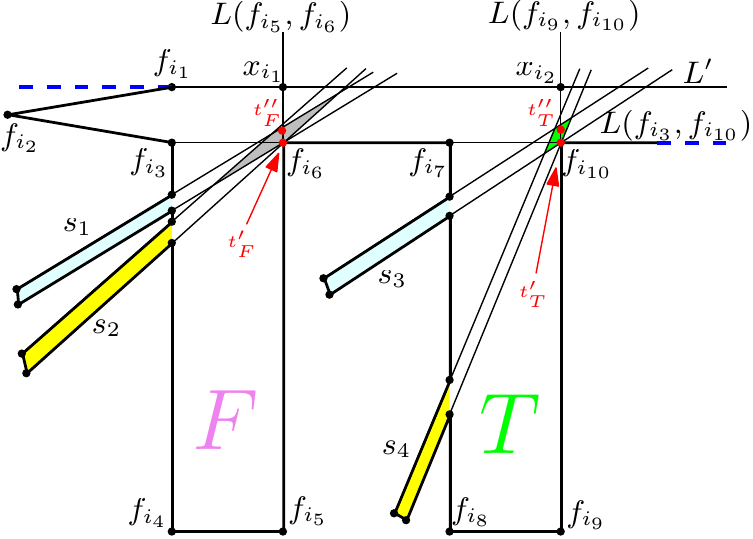}
\caption{Example of a variable pattern $P_i$ for the variable $u_i$.}
\label{fig:variable2}
\vspace{-10pt}
\end{wrapfigure}
It is important to notice that despite that $t_F'$, $t_F''$ are positioned on the boundary of $P_i^F$ and $t_T'$, $t_T''$ are positioned on the boundary of $P_i^T$, the towers $t_F''$ and $t_T''$ do not belong to the boundary of the main polygon. In Section~\ref{subsec:complete}, when we present the complete construction of the main polygon, we show that the ambiguities along the lines $L(t_F', t_F'')$ and $L(t_T', t_T'')$ during the trilateration can be avoided, because the map of $P$ together with the coordinates of all the towers is available to the agent. For now assume that the primary localization region of $t_F'$, $t_F''$ is to the left of $L(t_F', t_F'')$ and that the primary localization region of $t_T'$, $t_T''$ is to the left of $L(t_T', t_T'')$.

We showed how to position $4$ towers to trilaterate a variable pattern. However, only one of the two: $P_i^F$ or $P_i^T$ will require a pair of towers as described in the two previous paragraphs. No other tower will be needed to trilaterate the remainder of the variable pattern if all the literals for this variable are assigned truth values \textbf{consistently}. Consider an example of Fig.~\ref{fig:variable1}. Assume that $u_i$ is assigned the truth value \emph{true}. Then we position towers in $P_i^T$ only. We need to show that $P_i^F$ is trilaterated as well. All the literal patterns of the literals that equal to $u_i$ get a pair of towers in the ``true'' position, namely at and close to $a$-type vertices. Those towers can see inside the corresponding spikes in $P_i^F$. All the literal patterns of the literals that equal to $\overline{u_i}$ get a pair of towers in the ``false'' position, namely at and close to $x$-type vertices. Those towers also can see inside the corresponding spikes in $P_i^F$. Thus, all the spikes in the F-well of  $P_i^F$ are trilaterated. In addition, the towers of $P_i^T$ trilaterate $\{ f_{i_1}, f_{i_2}, f_{i_3}, f_{i_3}, x_{i_1}\}$ -- the subpolygon of $P_i^F$. The only subpolygon of $P_i^F$ that is not trilaterated is the quadrilateral $\{ f_{i_3}, f_{i_4}, f_{i_5}, f_{i_6}\}$. Assume now that $u_i$ is assigned the truth value \emph{false}. We position towers in $P_i^F$ only and need to show that $P_i^T$ is trilaterated too. Similarly to the previous case it can be shown that all the spikes in the T-well of $P_i^T$ are trilaterated. The towers of $P_i^F$ trilaterate $\{ f_{i_1}, f_{i_2}, f_{i_3}, f_{i_3}, x_{i_1}\}$ -- the subpolygon of $P_i^T$. The only subpolygon of $P_i^T$ that is not trilaterated is $\{ x_{i_1}, f_{i_6}, f_{i_7}, f_{i_8}, f_{i_9}, x_{i_2}\}$. In the following section we show a position of a pair of towers that trilaterate all the subpolygons $\{ f_{i_3}, f_{i_4}, f_{i_5}, f_{i_6}\}$ and $\{ x_{i_1}, f_{i_6}, f_{i_7}, f_{i_8}, f_{i_9}, x_{i_2}\}$ for all $1 \leq i \leq n$.

Our discussion can be summarized in the following lemma.
\begin{lemma}
\label{lem:variable}
Let $1 \leq i \leq n$, consider all the spikes of the variable pattern for $u_i$. If $u_i$ is assigned the truth value \emph{true} (respectively \emph{false}) then all the spikes inside the T-well (respectively, the F-well) are trilaterated by a pair of towers assigned to the variable pattern of $u_i$. All the spikes inside the F-well (respectively, the T-well) are trilaterated by pairs of towers assigned to literal patterns of $u_i$ or $\overline{u_i}$ in each clause junction that uses $u_i$.

\end{lemma}

Notice, that $\triangle f_{i_1} f_{i_2} f_{i_3}$ is trilaterated by a pair of towers assigned to the variable pattern.

\subsection{Complete Construction}
\label{subsec:complete}
\pdfbookmark[2]{Complete Construction}{subsec:complete}

We put variable patterns and clause junctions together as shown in Fig.~\ref{fig:complete}. This figure depicts an example of a complete polygon for a $3CNF$ formula $(u_1 \vee \overline{u_2} \vee \overline{u_3}) \wedge (u_1 \vee u_2 \vee \overline{u_3})$. This formula contains $n = 3$ variables and $m = 2$ clauses. Thus, the main polygon $P$ is comprised of $3$ variable patterns and $2$ clause junctions (each of which contains $4$ literal patterns). 

\begin{figure}[h]
\centering
\includegraphics[width=1\textwidth]{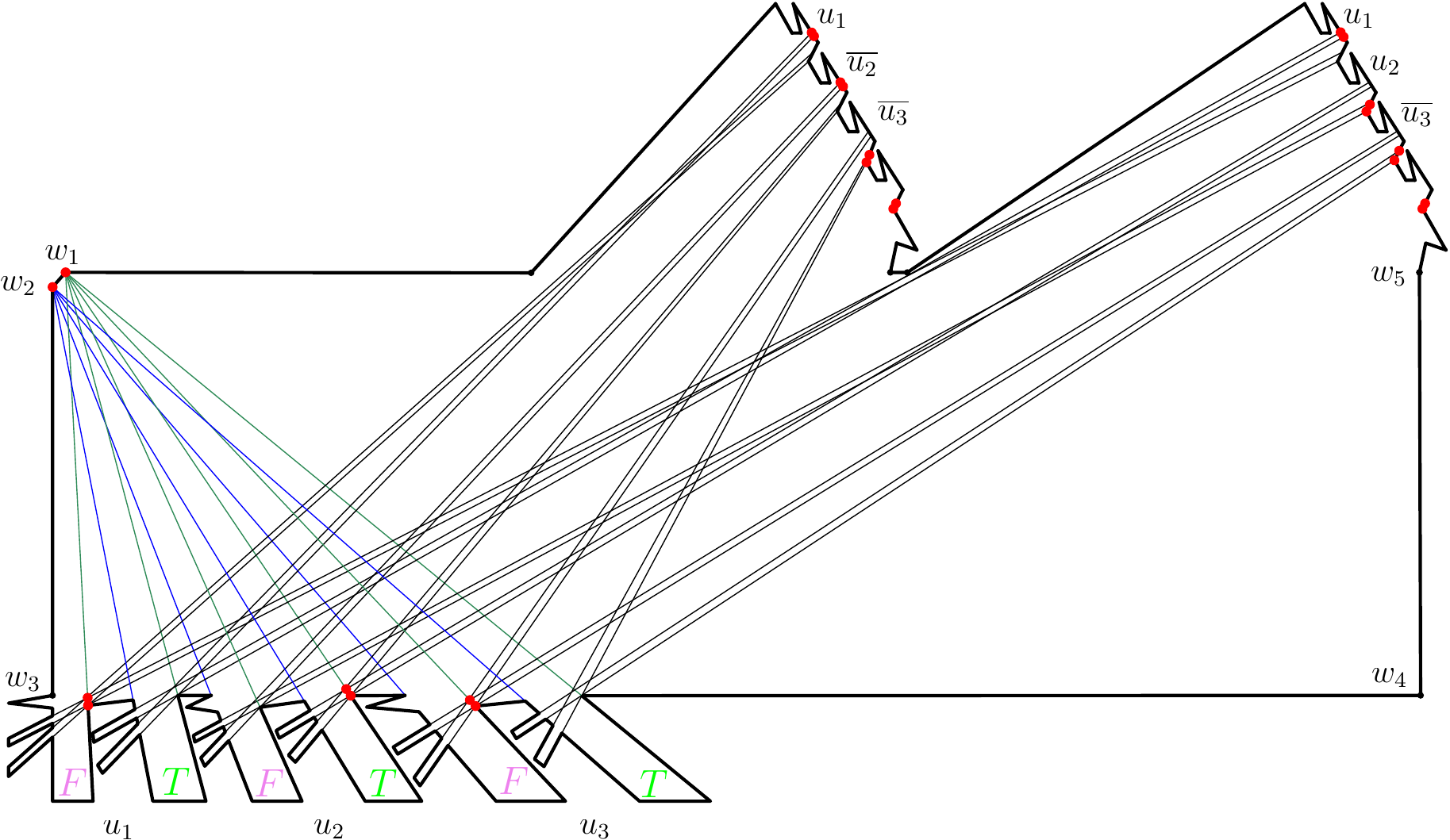}
\caption{The complete polygon $P$ for $(u_1 \vee \overline{u_2} \vee \overline{u_3}) \wedge (u_1 \vee u_2 \vee \overline{u_3})$ with an example of $8m + 2n + 2 = 24$ towers (shown in red) which is minimum. The tower positions correspond to the following truth assignments: $u_1$, $u_3$ are \emph{false}; $u_2$ is \emph{true}.}
\label{fig:complete}
\end{figure}

In Section~\ref{subsec:variable} we described how the consistency of the assigned truth values is enforced among all the literals of a specific variable via spike construction. We also proposed to position only a pair of towers per variable pattern. However, it remained unanswered how to trilaterate all the wells that weren't assigned towers. Now we are able to address the issue. We align the boundaries of the wells for all variable patterns with a pair of vertices of $P$. Refer to Fig.~\ref{fig:complete}. In particular, the left boundaries of both wells $\overline{f_{i_3} f_{i_4}}$ and $\overline{f_{i_7} f_{i_8}}$ for each variable pattern are aligned with the vertex $w_2$ (i.e. $w_2 \in L(f_{i_3}, f_{i_4})$ and $w_2 \in L(f_{i_7}, f_{i_8})$ for $1 \leq i \leq n$). Similarly, the right boundaries of both wells $\overline{f_{i_5} f_{i_6}}$ and $\overline{f_{i_9} f_{i_{10}}}$ for each variable pattern are aligned with the vertex $w_1$. Notice that the pair of vertices $w_1$, $w_2$ can see the interior of all the wells, and in particular, all the polygons $\{ f_{i_3}, f_{i_4}, f_{i_5}, f_{i_6}\}$ and $\{ x_{i_1}, f_{i_6}, f_{i_7}, f_{i_8}, f_{i_9}, x_{i_2}\}$ for all $1 \leq i \leq n$. Let $w$ be an intersection point of $L(f_{1_5}, f_{1_6})$ and $L(f_{n_7}, f_{n_8})$. Notice that $\triangle w w_1 w_3$ is the kernel of the polygon $\{w_1, w_2, w_3, w_4, w_5\} \cup \{$all the wells$\}$. In addition, $\triangle w w_1 w_3 \cap \partial P = \overline{w_1 w_2}$. We position a pair of towers: one at $w_1$ and another at $w_2$. Now all the wells and the polygon $\{w_1, w_2, w_3, w_4, w_5\}$ are trilaterated with a pair of towers. 

Recall, that in Section~\ref{subsec:variable} we positioned the towers $t_F''$ and $t_T''$ \textbf{not} on the boundary of $P$. We show that the ambiguities along the lines $L(t_F', t_F'')$ and $L(t_T', t_T'')$ can be avoided because the map of $P$ together with the coordinates of all the towers are available to the agent. Assume for example that an agent received messages from only two towers $t_F'$ and $t_F''$ of variable pattern for $u_i$ (refer to Fig.~\ref{fig:variable2},~\ref{fig:complete}). Thus, the agent can be in one of the two locations $p_1$ or $p_2$ that are reflections of each other along $L(t_F', t_F''$). We use the map of $P$ to calculate the visibility polygon of the pair of towers $V(t_F') \cap V(t_F'')$. If one of the locations, say $p_2$, does not belong to this visibility polygon, then the agent is at $p_1$. So, assume that both $p_1$ and $p_2$ belong to this visibility polygon. Observe that the agent cannot be in any clause junction $P_C$, otherwise it will be seen by at least a pair of towers positioned inside $P_C$, which contradicts to the fact that the agent is seen by $t_F'$ and $t_F''$ only. Recall, that there is a pair of towers at $w_1$ and $w_2$. The agent knows this information and, together with the map of $P$, can calculate the visibility polygon $V^* = ((V(t_F') \cap V(t_F'')) \setminus (V(w_1) \cap V(w_2))) \setminus \{$all the clause junctions$\}$. This visibility polygon consists only of the spikes of the F-well of variable pattern for $u_i$ and $\triangle f_{i_1} f_{i_2} f_{i_3}$. Moreover, $V^*$ belongs to one side of $L(t_F', t_F''$). Thus, if the agent is seen by $t_F'$ and $t_F''$ only, then it has a unique location at $V^*$. Similarly, if the agent is seen by $t_T'$ and $t_T''$ only then it has a unique location at $((V(t_T') \cap V(t_T'')) \setminus (V(w_1) \cap V(w_2))) \setminus \{$all the clause junctions$\}$.

\begin{lemma}
\label{lem:atleast}
At least $K = 8m + 2n + 2$ towers are required to trilaterate $P$.
\end{lemma}
\begin{proof}
We showed in Section~\ref{subsec:clause} that any trilateration of a clause junction requires at least $8$ towers. In Section~\ref{subsec:variable} we discussed that each variable pattern requires at least $2$ towers. In this section we showed that an additional $2$ towers are necessary to trilaterate $\{w_1, w_2, w_3, w_4, w_5\} \cup \{$all the wells$\}$. Since $P$ is the union of $m$ clause junctions, $n$ variable patterns and the polygon $\{w_1, w_2, w_3, w_4, w_5\}$, then its trilateration requires at least $K = 8m + 2n + 2$ towers.
\qed
\end{proof}

\begin{lemma}
\label{lem:minimum}
The minimum number of towers required to trilaterate the main polygon is $K = 8m + 2n + 2$ \textbf{if and only if} the given $3CNF$ formula is satisfiable.
\end{lemma}
\begin{proof}

\textbf{($\leftarrow$)} Assume that the given $3CNF$ formula is satisfiable. Then a truth assignment to the variables exists such that each of the clauses $C_j$, $1 \leq j \leq m$ has a truth value \emph{true}. By Lemma~\ref{lem:clause} every clause junction can be trilaterated with 8 towers. Lemma~\ref{lem:variable} implies that $2$ towers per variable pattern is sufficient to trilaterate all the spikes (that where not trilaterated by the towers in clause junctions) and $\triangle f_{i_1} f_{i_2} f_{i_3}$ for $1 \leq i \leq n$. We showed in this section that an additional $2$ towers trilaterate the remaining uncovered subpolygons of $P$ and resolve ambiguities. Thus $P$ can be trilaterated  with $K = 8m + 2n + 2$ towers, which by Lemma~\ref{lem:atleast} is the minimum number.

\textbf{($\rightarrow$)} Assume that the given $3CNF$ formula is \textbf{not} satisfiable. It follows that there exists at least one clause $C_j$ that has a truth value \emph{false} under any truth assignment. By Lemma~\ref{lem:clause} the clause junction $P_{C_j}$ needs at least $9$ towers. Similar to the proof of Lemma~\ref{lem:atleast} $P$ requires at least $8m + 2n + 3 = K + 1$ towers for trilateration. 
\qed
\end{proof}

\subsection{Construction takes Polynomial Time.}
\label{subsec:polynomial}
\pdfbookmark[2]{Construction takes Polynomial Time.}{subsec:polynomial}

To complete the proof of our main result -- Theorem~\ref{theo:NP}, we need to show that the reduction takes polynomial time. We have to demonstrate that the number of vertices of $P$ is polynomial in $n$ and $m$ and that the number of bits in the binary representation of the coordinates of those vertices is bounded by a polynomial in $n$ and $m$. During the reduction we create a simple polygon $P$ of size $49m + 10n + 3$. Every clause junction consists of $25$ vertices. Every literal creates two spikes, thus the total number of vertices in $P$ occupied by spikes is $24 m$. Every variable pattern (without spikes) consists of $10$ vertices. The final construction of $P$ includes $3$ more vertices for $w_1$, $w_2$ and $w_4$. Thus, the total number of vertices of $P$ is $25m + 24m + 10n + 3 = 49m + 10n + 3$, which is polynomial in $n$ and $m$. Notice that $n \leq 3m$ thus instead of saying ``polynomial in $n$ and $m$'' we can just use ``polynomial in $m$''.

Consider a construction $P_C$ of a clause $C$ (refer to Fig.~\ref{fig:clause}). Let $v_4'$ be an orthogonal projection of $v_4$ onto $L(v_5, v_1)$. We would like to fix the part of $P_C$ to the right of $L(v_4, v_4')$ and keep it identical among all the clauses of the given $3CNF$ formula. Assume that $v_4'$ is the origin of the coordinate system. We can change the positions of all the vertices of $P_C$ to the right of $L(v_4, v_4')$ (keeping collinearities and main features intact) such that the binary representation of their coordinates is polynomial in $m$. The relative position of the vertex $v_5$ to other vertices of $C$ will differ from clause to clause (refer to Fig.~\ref{fig:polynomial}).

\begin{figure}
\centering
\includegraphics[width=1\textwidth]{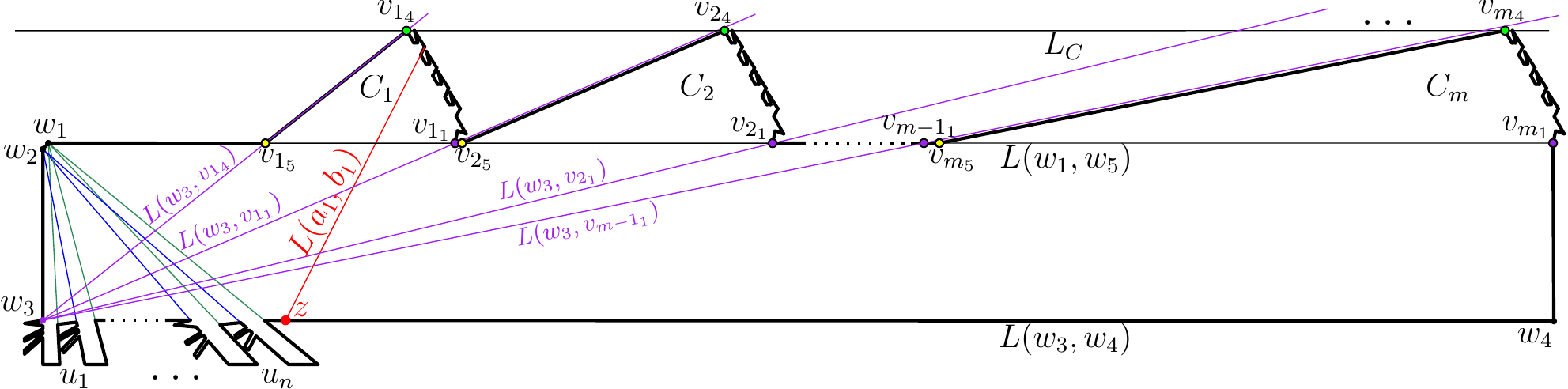}
\caption{Construction of $P$ is polynomial in $m$ and $n$. $w_5 = v_{m_1}$.}
\label{fig:polynomial}
\end{figure}


In the complete construction of $P$ we position $w_3$ at the origin of the coordinate system, i.e. the coordinates of $w_3$ are $(0,0)$. When we add all the variable patterns $P_{u_i}$, $1 \leq i \leq n$ to $P$ we can position every following variable pattern $P_{u_{i+1}}$ slightly below (by $y$-coordinate) $P_{u_i}$, otherwise the towers in $P_{u_{i+1}}$ could see $\triangle f_{i_1} f_{i_2} f_{i_3}$. Alternatively, we can rotate the line $L(f_{i_3}, f_{i_{10}})$ around $f_{i_3}$ to contain $x_{i_2}$ (refer to Fig.~\ref{fig:variable2}) and extend the boundaries of the wells such that $f_{i_{10}} = x_{i_2}$ and $f_{i_6}$, $f_{i_7} \in L(f_{i_3}, f_{i_{10}})$. We will use the latter approach. Now we can set the $y$-coordinate of the vertices $f_{1_1}$, $f_{1_{10}}$, $f_{2_1}$, $f_{2_{10}}$, $\ldots$ , $f_{n_1}$, $f_{n_{10}}$ to $0$ (note that $w_3 = f_{1_1}$).

Let $k_0, k_1, k_2, k_3, k_4, k_5, k_6$ be an arbitrary sorted sequence of positive integers of polynomial size in $m$, such that $k_6 < 2 k_5$. We set the $x$-coordinate of $w_1$ to be $k_1$ and its $y$-coordinate to be $k_5$. In other words, $w_1 = (k_1, k_5)$. The $y$-coordinate of $w_2$ is set to $k_4$, i.e. $w_2 = (0,k_4)$. We set the distance between the vertices $f_{i_1}$ and $f_{i_{10}}$ for $1 \leq i \leq n$ to be $3k_2$ and the distance between $f_{i_{10}}$ and $f_{{i+1}_1}$ for $1 \leq i < n$ to be $k_2$. Refer to Figures~\ref{fig:variable2} and~\ref{fig:polynomial}. Recall that the $y$-coordinate of $f_{i_1}$, $f_{i_{10}}$ for $1 \leq i \leq n$ is $0$. We just defined the vertices $f_{i_1} = (4k_2(i-1),0)$ and $f_{i_{10}} = (3k_2 + 4k_2(i-1),0)$ for $1 \leq i \leq n$. Notice that all the coordinates defined so far are polynomial in $m$.

The $y$-coordinate of $f_{i_3}$ for $1 \leq i \leq n$ is set to $- k_0$. To determine its $x$-coordinate observe that $f_{i_3}$, $1 \leq i \leq n$ is an intersection point between $L(w_2, f_{i_1})$ and the horizontal line through $(0, -k_0)$. To proceed further recall the following.

The intersection point of two lines $L_1$ (defined by two distinct points $(x_1, y_1)$ and $(x_2, y_2)$) and $L_2$ (defined by two distinct points $(x_3, y_3)$ and $(x_4, y_4)$) can be written out as

\[(x,y)= \left( \frac{(x_1y_2-y_1x_2)(x_3-x_4) - (x_1-x_2)(x_3y_4-y_3x_4)}{(x_1-x_2)(y_3-y_4)-(y_1-y_2)(x_3-x_4)}, 
\frac{(x_1y_2-y_1x_2)(y_3-y_4) - (y_1-y_2)(x_3y_4-y_3x_4)}{(x_1-x_2)(y_3-y_4)-(y_1-y_2)(x_3-x_4)} \right).
\] 

This computation yields a fraction whose numerator is a polynomial of degree three in the input coordinates and whose denominator is a polynomial of degree two in the input coordinates. That is, if the input coordinates are $b$-bit numbers then the output coordinates need at most $5b$ bits to be represented. The proof of this can be found in lecture notes in Computational Geometry by M. Hoffmann~\cite{MichaelHoffmann2009}.

Thus, the $x$-coordinate of $f_{i_3}$, $1 \leq i \leq n$ can be determined via the above computation. Moreover, the number of bits required for its $x$-coordinate is at most $5$ times the number of bits used for $f_{i_1}$ (assuming $f_{i_1}$ had the most bits in its representation among $w_3$, $f_{i_1}$, $(0, -k_0)$ and $(1, -k_0)$), and it is still polynomial in $m$.

In a similar way we find the coordinates of other vertices that define wells. The vertex $f_{i_6}$ for $1 \leq i \leq n$ is an intersection of $L(w_1, (k_2+4k_2(i-1),0))$ and $L(f_{i_3}, f_{i_{10}})$; and the vertex $f_{i_7}$ for $1 \leq i \leq n$ is an intersection of $L(w_2, (2k_2+4k_2(i-1),0))$ and $L(f_{i_3}, f_{i_{10}})$. We set the $y$-coordinate of all the vertices at the bottom of every well to $- k_3$. Thus the $x$-coordinate of $f_{i_4}$ (respectively $f_{i_5}$, $f_{i_8}$ and $f_{i_9}$) for $1 \leq i \leq n$ is equal to the $x$-coordinate of the intersection point between the horizontal line through $(0,- k_3)$ and $L(w_2, f_{i_1})$ (respectively $L(w_1, (k_2+4k_2(i-1),0))$, $L(w_2, (2k_2+4k_2(i-1),0))$ and $L(w_1, f_{i_{10}})$). Notice that the number of bits required to represent the coordinates of each of the discussed vertices is polynomial in $m$ and is at most $5$ times the number of bits required to represent $f_{i_1}$ or $f_{i_{10}}$. The vertex $f_{i_2}$, $1 \leq i \leq n$ does not require precise construction; there is plenty of room to choose polynomial coordinates for it. 

\medskip
We proceed now with the construction of the clause junctions and after that we will return to discuss the spike formation inside the variable patterns.

Consider the first clause $C_1$ in the given $3CNF$ formula. Let $P_{C_1}$ be its clause junction. Let $z$ be an intersection point of the two lines $L(a_1, b_1)$ and $L(w_3, w_4)$ (refer to Figures~\ref{fig:clause} and~\ref{fig:polynomial}). We position $P_{C_1}$ in $P$ in such a way that the $x$-coordinate of $z$ is not smaller than the $x$-coordinate of $f_{n_{10}}$. Otherwise, we cannot guarantee that the towers of every literal pattern can see the vertices $f_{i_1}$, $f_{i_{10}}$ for $1 \leq i \leq n$. By construction the angle that $L(a_1, b_1)$ creates with the positive direction of $X$-axis is bigger than $\pi/3$. Thus if we set the $x$-coordinate of $v_{1_4}$ (where $v_{1_4}$ is the vertex $v_4$ of the first clause $C_1$) to be twice bigger than the $x$-coordinate of $f_{n_{10}}$, then $z$ is guaranteed to be to the right of $f_{n_{10}}$. Hence, we set $v_{1_4} = (2k_2(4n-1), k_6)$ (note that $f_{n_{10}} = (3k_2 + 4k_2(n-1),0) = (k_2(4n-1),0)$). Recall that in the beginning of this subsection we fixed the part of $P_{C_j}$ to the right of $L(v_{j_4}, v_{j_4}')$ for $1 \leq j \leq m$ and assigned polynomial coordinates to all the vertices of $P_{C_j}$ except for $v_{j_5}$ (assuming that $v_{j_4}$ is the origin). To add $P_{C_1}$ to $P$ we just translate all the vertices of $P_{C_1}$ (except for $v_{1_5}$) by vector $(2k_2(4n-1),k_5)$. The $y$-coordinate of $v_{1_5}$ is $k_5$. Its $x$-coordinate should be small enough for the towers inside the first literal of $C_1$ to see $f_{1_3}$. Since $k_5 > k_6 / 2$ there are plenty of polynomial choices for $v_{1_5}$. Notice that the number of bits in binary representation of the coordinates of all the vertices of $P_{C_1}$ is bounded by polynomial in $m$.

To add the clause junction $P_{C_2}$ of the second clause $C_2$ to $P$ we consider the line $L(w_3, v_{1_1})$. Refer to Fig.~\ref{fig:polynomial}. Let $h$ be the intersection point between $L(w_3, v_{1_1})$ and the horizontal line through $(0,k_6)$. The $y$-coordinate of $v_{2_4}$ is set to $k_6$, and its $x$-coordinate should be to the right of $h$. Since $k_5$ is bigger than $k_6-k_5$, the length of the line segment $\overline{w_3 v_{1_1}}$ is bigger then the length of $\overline{v_{1_1} h}$. Hence, if we set the $x$-coordinate of $v_{2_4}$ to be twice bigger then the $x$-coordinate of $v_{1_1}$ then $v_{2_4}$ is guaranteed to be to the right of $h$. Moreover, the number of bits of the $x$-coordinate of $v_{2_4}$ will be larger then the number of bits of the $x$-coordinate of $v_{1_1}$ by at most $1$. All other vertices of $P_{C_2}$ can be added similarly to the vertices of $P_{C_1}$ (via the translation by the vector $(2\{x$-coordinate of $v_{1_1}\},k_5)$). The vertex $v_{2_4}$ can be positioned close to $v_{1_1}$ or even at $v_{1_1}$.

We keep on adding clause junctions $P_{C_j}$ for $1 \leq j \leq m$ to $P$ in a similar manner. The number of bits that represent coordinates of $P_{C_j}$ grows by at most $2$ compared to $P_{C_{j-1}}$. It follows that $v_{m_1}$ requires at most $2m$ bits more than $v_{1_4}$, which is still polynomial in $m$. We set the $x$-coordinate of $w_4$ to be equal to the $x$-coordinate of $v_{m_1}$. The $y$-coordinate of $w_4$ is $0$.

\medskip
We are ready to discuss the construction of the spikes. Let us define an arrangements of lines as follows:
\begin{itemize}
\item[$\bullet$] For every clause $C_j$, $1 \leq j \leq m$ that contains the literal $l_q$, $1 \leq q \leq 3$ equal to the variable $u_i$, $1 \leq i \leq n$ create a pair of lines $L(x_{j_q}, f_{i_{10}})$ and $L(a_{j_q}, f_{i_6})$.
\item[$\bullet$] For every clause $C_j$, $1 \leq j \leq m$ that contains the literal $l_q$, $1 \leq q \leq 3$ equal to $\overline{u_i}$, $1 \leq i \leq n$ (the negation of the variable $u_i$) create a pair of lines $L(x_{j_q}, f_{i_6})$ and $L(a_{j_q}, f_{i_{10}})$.
\item[$\bullet$] For every variable $u_i$, $1 \leq i \leq n$ create a line $L(f_{i_1}, f_{i_2})$.
\end{itemize}

Notice that $x_{j_q}$, for $1 \leq j \leq m$, $1 \leq q \leq 3$ is an intersection point between $L(e_{j_q}, d_{j_q})$ and $L(c_{j_q}, b_{j_q})$. Both lines defined by the points with polynomial coordinates. Thus the number of bits required to represent the coordinates of $x_{j_q}$ is at most $5$ times bigger then the number of bits in the coordinate representation of any of the four points. Thus the coordinates of $x_{j_q}$ are also polynomial in $m$.

Consider the T-well of the variable pattern for $u_i$, $1 \leq i \leq n$ (refer to Fig.~\ref{fig:complete}). Every line of the line arrangement that contains $f_{i_{10}}$ will define a spike in this well. We assume that a variable cannot participate in the same clause more than once. Thus the number of spikes in a particular well can be at most $m$. We need to define a close neighbourhood in $P$ around $f_{i_{10}}$ from where the interior of every spike of the well will be visible. Extend the line segment $\overline{f_{i_9} f_{i_{10}}}$ towards $w_1$ until it first time hits a line in the arrangement (let $f'$ be that intersection point). Let  $f''$ be a midpoint of the line segment $\overline{f' f_{i_{10}}}$. It is important to emphasize that we do not calculate lengths of line segments as a Euclidean distance between the points, because is involves the usage of a square root. To find the coordinates of $f''$ we first find the coordinates of $f'$ (via the intersection of two lines defined by $4$ points with polynomial coordinates). The $y$-coordinate of $f''$ is a half of the difference between the $y$-coordinate of $f'$ and the $y$-coordinate of $f_{i_{10}}$. The $x$-coordinate of $f''$ is found in a similar way. 

We can choose $m$ different points on the line segment $\overline{f'' f_{i_{10}}}$ whose coordinates are polynomial in $m$. For every line $L$ in the arrangement that contains $f_{i_{10}}$ we do the following:

\begin{itemize}
\item[$\bullet$] If $L$ contains $a_{j_q}$ then create a line $L^*$ parallel to $L$ that passes through one of the $m$ points on $\overline{f'' f_{i_{10}}}$ and intersects the line segment $\overline{a_{j_q} b_{j_q}}$.
\item[$\bullet$] If $L$ contains $x_{j_q}$ then rotate $L$ clockwise around $x_{j_q}$ to contain one of the $m$ points on $\overline{f'' f_{i_{10}}}$. Create a line $L^*$ parallel to the updated $L$ such that it passes through $f_{i_{10}}$ and intersects the line segment $\overline{x_{j_q} b_{j_q}}$.
\end{itemize}

Every corresponding pair of lines $L$, $L^*$ defines a strip that in turn defines a spike. We choose the points on $\overline{f'' f_{i_{10}}}$ to define the pairs $L$, $L^*$ in such a way that no two different strips overlap to the left of $L(f_{i_7}, f_{i_8})$. 

Notice that $L$ is a line defined by a pair of points with polynomial coordinates; and $L^*$ is a translation of $L$ by a vector with polynomial coordinates. Thus $L^*$ can also be defined via a pair of points with polynomial coordinates. Every spike contains $4$ vertices. The coordinates of the two vertices of the spike that belong to $\overline{f_{i_7} f_{i_8}}$ are determined as the intersection points between $L(f_{i_7}, f_{i_8})$, $L$ and $L^*$. There are many polynomial choices in the intersection of the strip $L(f_{i_7}, f_{i_8})$, $L(f_{i_6}, f_{i_5})$ with $L$ and $L^*$ for the other two vertices of the spike.

In a similar way we define the spikes in the $F$-well of the variable pattern for $u_i$, $1 \leq i \leq n$. Notice that the coordinates of all the spikes in $P$ are polynomial in $m$.

\medskip
We showed the construction of a simple polygon $P$ that has $49m + 10n + 3$ vertices such that the binary representation of the coordinates of those vertices is bounded by a polynomial in $n$ and $m$. Together with Lemma~\ref{lem:minimum} this concludes the proof of Theorem~\ref{theo:NP}.

\section{Concluding Remarks}
\label{sec:conclusion}
\pdfbookmark[1]{Concluding Remarks}{sec:conclusion}

We showed that the problem of determining the minimum number of broadcast towers that can localize a point anywhere in $P$ is NP-hard even if the complete information about the polygon and the coordinates of all the towers are available to the point that wants to locate itself in the art gallery. To prove NP-hardness we showed a reduction from a known NP-complete problem $3SAT$ to our problem. We proved that the reduction takes polynomial time by showing that the number of vertices of the constructed polygon is polynomial in the size of the input and that the number of bits in the binary representation of the coordinates of those vertices is bounded by a polynomial in the size of the input. 

An obvious open question is whether the AGL problem is NP-complete. Assume that the vertices of the input polygon are rational numbers. The input then is represented by a finite bitstring of length $N$. The complexity measure is the number of bit-operations, as a function of $N$. To show that AGL is in NP, we must show the following: 

If the input polygon $P$ can be trilaterated by at most $K$ towers, then there exists a certificate $C$, such that: 

\begin{enumerate}
\item the number of bits to represent $C$ is polynomial in $N$.
\item Given $P$, $K$, and $C$, we can verify in time polynomial in $N$ that $C$ is a correct certificate. 
\end{enumerate}

A certificate would be the locations of the $K$ towers. The question then becomes if the total number of bits to represent these locations is polynomial in $N$. We do not know if this is the case. But we also do not know if AGL is not in NP.

%
%

\bibliography{NP-hard}

\begin{thebibliography}{1}

\bibitem{AGLproblem17}
P.~Bose, J-L. de~Carufel, A.~Shaikhet, and M.~Smid.
\newblock Art gallery localization.
\newblock Submitted to the Journal of Computational Geometry: Theory and
  Applications (CGTA), June 2017.

\bibitem{Chvatal197539}
V.~Chv{\'a}tal.
\newblock A combinatorial theorem in plane geometry.
\newblock {\em Journal of Combinatorial Theory, Series B}, 18(1):39 -- 41,
  1975.

\bibitem{Fisk1978374}
S.~Fisk.
\newblock A short proof of {Chv{\'a}tal's} watchman theorem.
\newblock {\em Journal of Combinatorial Theory, Series B}, 24(3):374, 1978.

\bibitem{Garey:1979:CIG:578533}
M.~R. Garey and D.~S. Johnson.
\newblock {\em Computers and Intractability: A Guide to the Theory of
  NP-Completeness}.
\newblock W. H. Freeman \& Co., New York, NY, USA, 1979.

\bibitem{MichaelHoffmann2009}
M.~Hoffmann.
\newblock Lecture notes in computational geometry, October 2009.

\bibitem{Lee:1986:CCA:13643.13657}
D.~T. Lee and A.~K. Lin.
\newblock Computational complexity of art gallery problems.
\newblock {\em IEEE Trans. Inf. Theor.}, 32(2):276--282, March 1986.

\bibitem{O'Rourke:1987:AGT:40599}
J.~O'Rourke.
\newblock {\em Art Gallery Theorems and Algorithms}.
\newblock Oxford University Press, Inc., New York, NY, USA, 1987.

\end{thebibliography}

\end{document}